\documentclass{article}

\usepackage[pdftex]{graphicx} \pdfcompresslevel=9

\usepackage{flushend}
\usepackage{cite}
\usepackage{overpic}

\usepackage{amsmath,amssymb}
\usepackage{algorithm2e}

\title{Sparse Modeling of Intrinsic Correspondences}
\author{J. Pokrass$^{1}$, A. M. Bronstein$^{1}$, M. M. Bronstein$^{2}$, \\
       P. Sprechmann$^{3}$, and G. Sapiro$^{1}$
        \\
         $^1$School of Electrical Engineering, Tel Aviv University\\
         $^2$Faculty of Informatics, Universit\`{a} della Svizzera Italiana\\
         $^3$Department of Electrical and Computer Engineering, Duke University
       }

\usepackage{amsthm}

\newtheorem{proposition}{Proposition}
\newtheorem{lemma}{Lemma}

\begin{document}

\newcommand{\bb}[1]{\pmb{\mathrm{#1}}}
\newcommand{\Tr}{\mathrm{T}}

\newcommand{\prox}{\bb{P}}

\newcommand{\pp}{\bb{p}}
\newcommand{\qq}{\bb{q}}
\newcommand{\aaa}{\bb{a}}
\newcommand{\bbb}{\bb{b}}
\newcommand{\mm}{\bb{m}}
\newcommand{\nn}{\bb{n}}
\newcommand{\uu}{\bb{u}}
\newcommand{\vv}{\bb{v}}
\newcommand{\ww}{\bb{w}}
\newcommand{\xx}{\bb{x}}
\newcommand{\yy}{\bb{y}}
\newcommand{\ff}{\bb{f}}
\newcommand{\gggg}{\bb{g}}

\newcommand{\mmu}{\bb{\mu}}
\newcommand{\nnu}{\bb{\nu}}
\newcommand{\rrho}{\bb{\rho}}
\newcommand{\ssigma}{\bb{\sigma}}
\newcommand{\ttau}{\bb{\tau}}
\newcommand{\ttheta}{\bb{\theta}}
\newcommand{\ppi}{\bb{\pi}}
\newcommand{\bbeta}{\bb{\beta}}

\newcommand{\PPi}{\bb{\Pi}}
\newcommand{\PPsi}{\bb{\Psi}}
\newcommand{\PPhi}{\bb{\Phi}}

\newcommand{\ones}{\bb{1}}
\newcommand{\zeros}{\bb{0}}

\newcommand{\Ee}{\bb{E}}
\newcommand{\Pp}{\bb{P}}
\newcommand{\Qq}{\bb{Q}}
\newcommand{\Aa}{\bb{A}}
\newcommand{\Bb}{\bb{B}}
\newcommand{\Cc}{\bb{C}}
\newcommand{\Mm}{\bb{M}}
\newcommand{\Nn}{\bb{N}}
\newcommand{\Xx}{\bb{X}}
\newcommand{\Yy}{\bb{Y}}
\newcommand{\Vv}{\bb{V}}
\newcommand{\Gg}{\bb{G}}
\newcommand{\Ii}{\bb{I}}
\newcommand{\Dd}{\bb{D}}
\newcommand{\Ww}{\bb{W}}
\newcommand{\Rr}{\bb{R}}
\newcommand{\Ss}{\bb{S}}
\newcommand{\Ll}{\bb{L}}
\newcommand{\Oo}{\bb{O}}

\newcommand{\mypara}[1]{{\noindent \bf{#1.} }}

\newcounter{ALC@tempcntr}
\newcommand{\LCOMMENT}[1]{%
    \setcounter{ALC@tempcntr}{\arabic{ALC@rem}}
    \setcounter{ALC@rem}{1}
    \item \{#1\}
    \setcounter{ALC@rem}{\arabic{ALC@tempcntr}}
}%
\newcommand{\BREAKLINE}{%
    \setcounter{ALC@tempcntr}{\arabic{ALC@rem}}
    \setcounter{ALC@rem}{1}
    \item
    \setcounter{ALC@rem}{\arabic{ALC@tempcntr}}
}%

\maketitle
\begin{abstract}

We present a novel sparse modeling approach to non-rigid shape matching using only the ability to detect repeatable regions. As the input to our algorithm, we are given only two sets of regions in two shapes; no descriptors are provided so the correspondence between the regions is not know, nor we know how many regions correspond in the two shapes.
We show that even with such scarce information, it is possible to establish very accurate correspondence between the shapes by using methods from the field of sparse modeling, being this, the first non-trivial use of sparse models in shape correspondence.
We formulate the problem of {\em permuted sparse coding}, in which we solve simultaneously for an unknown permutation ordering the regions on two shapes and for an unknown correspondence in functional representation. We also propose a robust variant capable of handling incomplete matches.
Numerically, the problem is solved efficiently by alternating the solution of a linear assignment and a sparse coding problem.
The proposed methods are evaluated qualitatively and quantitatively on standard benchmarks containing both synthetic and scanned objects.

\end{abstract}

\section{Introduction}

Matching of deformable shapes is a notoriously difficult problem playing an important role in many applications \cite{kaick2010survey}. 
Unlike rigid matching where the correspondence can be parametrized by a small number of parameters (rotation and translation of one shape w.r.t. the other \cite{ChenMedioni:91:ICP,BeslMcKay:92:ICP}), non-rigid matching  typically uses point-wise representation of correspondence, which results in the number of degrees of freedom growing exponentially with the number of matched points.

Non-rigid correspondence methods try to find correspondence by minimizing some structure distortion. 
The structures can be point-wise (local descriptors \cite{zaharescu-surface,sunHKS,gebal:andreas:etal:adf:09,aubry2011wave}), pair-wise (distances \cite{ela:kim:FLATTEN,Memoli05,bro:bro:kim:PNAS,bro-ghf}), or higher order \cite{zeng2010dense}.

In order to make the matching problem computationally feasible, it is crucial to reduce the size of the search space \cite{tevs2011intrinsic}. 
%
%
Most methods use a combination of point- and pair-wise structure matching in order to achieve this, and  
%
%
typically consist of three main components: feature detection, feature description, and regularization. 
Given two shapes, a {\em feature detector} allows to find a set of landmarks (points or regions) that are repeatable, i.e., appear (possibly with some inaccuracy) on both shapes. 
{\em Feature descriptor} then assigns to each feature a vector capturing some local geometric properties of the shape; very often, the two processes are combined into a single one. 
Using the descriptors, landmarks on two shapes can be matched (it has been shown \cite{ovsjanikov2010one} that under some conditions, correct landmark matching fully determines the intrinsic correspondence between the shapes). 
Such a matching reduces the search space size to points with similar descriptors. 
However, since the matching uses only local information, such correspondence can be noisy, and some kind of {\em regularization} based on higher-order information is needed to rule out bad or inconsistent correspondences. 
This information is also used to establish the correspondence between the rest of the points on the shapes. 
Often, the process is applied hierarchically, restricting the candidate matches to points in the proximity of the landmarks \cite{sahillioglu2012}. 

Computer graphics and geometry processing literature contains a plethora of approaches for each of the aforementioned components. 
Feature detection methods try to locate stable points or regions \cite{digne2010level,litman:BB:11} that are invariant under isometric deformations and robust to noise.  
Popular feature descriptors include the heat kernel signature (HKS) \cite{sunHKS,gebal:andreas:etal:adf:09},  
wave kernel signature (WKS) \cite{aubry2011wave}, global point signature (GPS) \cite{Rustamov07} or methods adopted from the domain of image analysis \cite{zaharescu-surface}.
As regularization, pairwise structures such as geodesic \cite{Memoli05,bro:bro:kim:PNAS} or diffusion distances \cite{bro-ghf} and higher-order structures \cite{zeng2010dense} have been used.

Alternatively, there have been several attempts to represent correspondences with a small set of parameters.  
Elad and Kimmel \cite{ela:kim:FLATTEN} 
used multidimensional scaling (MDS)-type methods to embed the intrinsic structure of the shapes into a low-dimensional Euclidean space, posing the problem of non-rigid matching as a problem of rigid matching of the corresponding embeddings (``canonical forms''). 
Mateus et al. \cite{Mateus08} used spectral embeddings instead of MDS. 
Lipman and Funkhouser \cite{Lipman09} embedded the shapes into a disk by means of conformal maps and represented the correspondence as a M{\"o}bius transformation.


More recently, Ovsjanikov et al. \cite{ovsjanikov2012functional} introduced the functional representation of correspondences, allowing to perform a ``calculus'' of correspondences. 
In this approach, correspondence is modeled as a correspondence between functions on two shapes rather than points, and can be compactly represented in the Laplace-Beltrami eigenbasis \cite{levy2006laplace} as a matrix of coefficients of decomposition of the basis functions of the first shape in the basis of the second one.  
In this paper, we will be relying upon this latter representation. 
%

\subsection{Main contribution}

The main practical contribution of this paper is an approach for finding dense intrinsic correspondence between near-isometric shapes with very little known information: we only assume to be able to detect regions in two shapes in a repeatable enough way (i.e., that at least some regions in one shape correspond accurately enough to some other regions in another shape). No region descriptors are given, so the correspondence of the regions is unknown. 
The assumption of near-isometry assures that in the functional representation of \cite{ovsjanikov2012functional}, the unknown correspondence can be represented as a sparse matrix. The assumption of repeatable regions implies that there exists some unknown permutation that orders the regions according to their correspondence.

We formulate the problem of {\em permuted sparse coding}, in which we simultaneously look for the permutation and the correspondence, thereby introducing the very successful area of sparse modeling into efficient and state-of-the-art shape correspondence. 
We note that with the permutation fixed, our problem becomes the standard sparse coding problem; having the correspondence fixed, the problem becomes a linear assignment. 
This allows efficient numerical solution by alternating the two aforementioned problems and employing efficient solvers that exist for both.

%

Our method relies on a pretty common assumption that the shapes are nearly-isometric (though our experimental results show our approach still works even when departing from this assumption), and out of all methods we are aware of, it uses perhaps the scarcest amount of data to establish dense correspondence between the shapes. 
%
For example, sandard region detectors with high repeatability such as \cite{litman:BB:11} are sufficient.

Compared to recent techniques for region-wise shape matching (see, e.g., \cite{golovinskiy2009consistent,van2011prior,huang2011joint,pokrass2012correspondence}), our approach has several important practical advantages: First, we do not use any feature descriptor. Second, most region-wise correspondence approaches require an additional step of extending the correspondence between matched regions to the rest of the points.


The rest of the paper is organized as follows. 
In Section~2, we overview the functional representation of correspondences, allowing to work with correspondences as algebraic structures, and state the main notions in sparse modeling. 
In Section~3, we formulate our problem of permuted sparse coding for establishing correspondence from a set of repeatable regions given in unknown order. We then extend the problem to the general setting where the region detection process is not perfectly repeatable. 
In Section~4, we describe the numerical optimization used to solve our permuted sparse coding  problem.  
Experimental results are shown in Section~5. Finally, Section~6 discusses the limitations and possible extensions of the proposed framework and concludes the paper.

\begin{figure*}[t!]
   \begin{center}
      \begin{overpic}
	[width=1\linewidth,trim=0 140 0 0,clip=true]{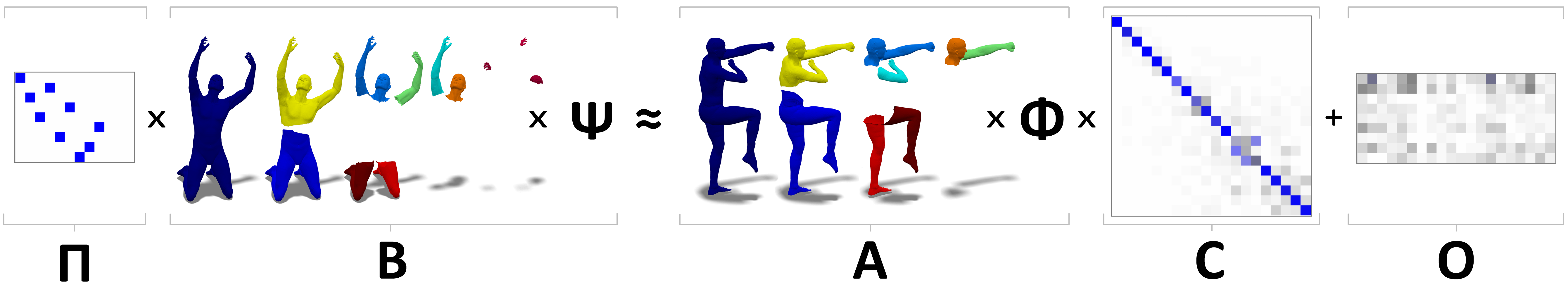}
	\put(4,-1.5){$\bb{\Pi}$}
	\put(22,-1.5){$\bb{B}$}
	\put(52.5,-1.5){$\bb{A}$}
	\put(76.5,-1.5){$\bb{C}$}
	\put(92.25,-1.5){$\bb{O}$} 		
\end{overpic}\vspace{3mm}
   \caption{ Near isometric shape correspondence as a sparse modeling problem (see details in text): Indicator functions
  of repeatable regions on two shapes are detected and represented as matrices of coefficients $\bb{A}$ and $\bb{B}$
  in the corresponding orthonormal harmonic bases $\bb{\Phi}$ and $\bb{\Psi}$. 
   When the regions 
   are brought into correspondence, the point-to-point correspondence between the shapes can be encoded by an approximately
   diagonal matrix $\bb{C}$. In the proposed procedure termed as \emph{permuted sparse coding}, we solve $\bb{\Pi}\bb{B} = \bb{A}\bb{C} + \bb{O}$ simultaneously for an approximately diagonal $\bb{C}$ and the permutation $\bb{\Pi}$ bringing the indicator functions into correspondence.
   To cope with imperfectly matching regions, we relax the surjectivity of the permutation and absorb the mismatches into a row-wise sparse outlier matrix $\bb{O}$. For visualization purposes, the coloring of the regions is
   consistent as after the application of the permutation. 
   Correspondence is shown between a synthetic TOSCA and scanned SCAPE shape.
   }
   \label{fig:teaser}
   \end{center}
\end{figure*}

\section{Background}

\begin{figure*}[t!]
   \begin{center}
   \includegraphics[width=\linewidth]{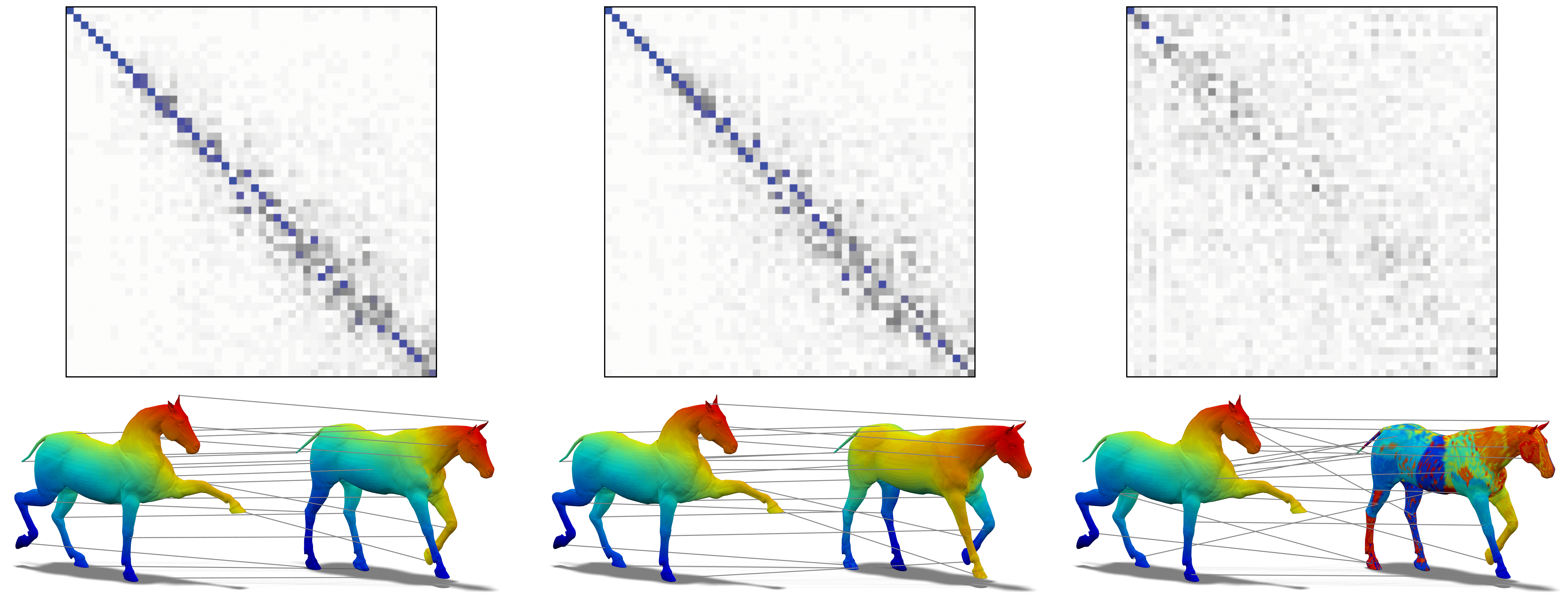}
   \caption{ 
    Top row: representation of different maps between the two shapes using the matrix $\bb{C}$. Shown left-to-right are
    the ideal correspondence, a symmetric correspondence, and a random correspondence.
    Bottom row: representation of the same correspondences as point-to-point maps. Note that the farther is the correspondence from an isometry,
    the less diagonal is the matrix $\bb{C}$.
    }
   \label{fig:funcorr}
   \end{center}
\end{figure*}

\begin{figure*}[t!]
   \begin{center}
   \includegraphics[width=1\linewidth]{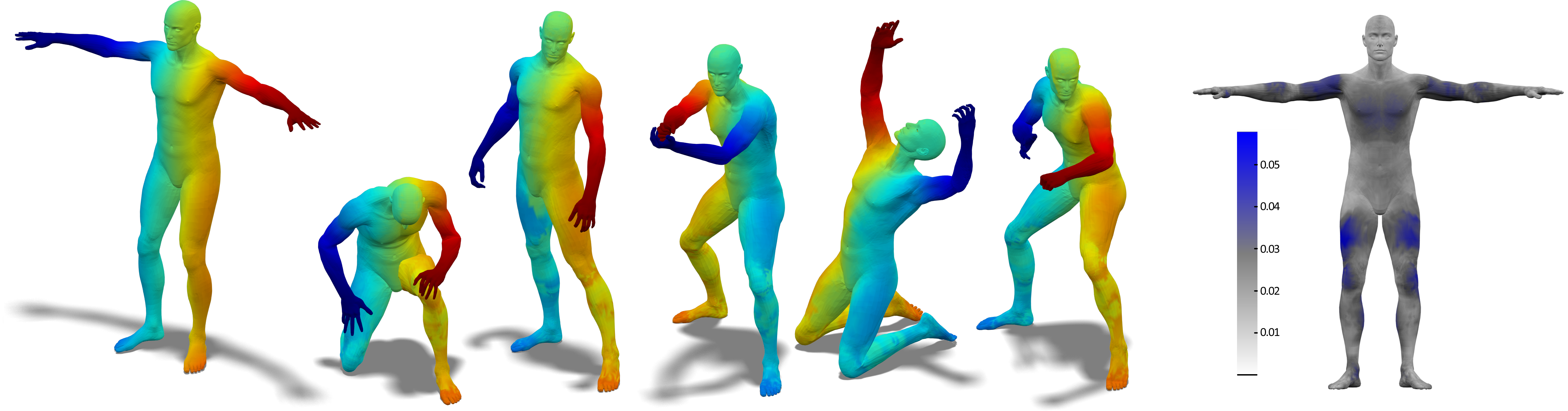}
   \caption{Dense point-to-point correspondences obtained between the left TOSCA human shape and its approximate isometries. Corresponding points are marked with consistent colors.
   The average correspondence distortion is depicted
   in units of the reference shape diameter. The highest distortions are obtained on the non-isometric joints, but do not exceed $6\%$ of the diameter. }
   \label{fig:tosca}
   \end{center}
\end{figure*}



\subsection{Functional representation of correspondences}

The direct representation of correspondences as maps between two non-Euclidean spaces limits the range of tools that can be employed for correspondence computation
due to the lack of an algebraic structure. In this paper, we rely on the functional representation of correspondences introduced in \cite{ovsjanikov2012functional}, which overcomes this limitation. In what follows, we briefly review the main idea of such functional representations.

Let $X$ and $Y$ be two shapes, modeled as compact smooth Riemannian manifolds,
related by a bijective correspondence $t : X \rightarrow Y$.
Then, for any real function $f : X \rightarrow \mathbb{R}$, we can construct a corresponding function
$g : Y \rightarrow \mathbb{R}$ as $g = f \circ t^{-1}$.
The correspondence $t$ uniquely defines a mapping between two function spaces $T : \mathcal{F}(X,\mathbb{R}) \rightarrow \mathcal{F}(Y,\mathbb{R})$,
where $\mathcal{F}(X,\mathbb{R})$ denotes the space of real functions on $X$.
Such a representation is linear, since for every pair of functions $f_1, f_2$ and scalars $\alpha_1, \alpha_2$,
\begin{align}
T(\alpha_1f_1+\alpha_2f_2) &= (\alpha_1f_1 + \alpha_2f_2) \circ t^{-1}
\nonumber\\               &= \alpha_1f_1 \circ t^{-1}+\alpha_2f_2\circ t^{-1}
\nonumber\\               &= \alpha_1T(f_1) + \alpha_2T(f_2).
\end{align}

Assuming that $X$ is equipped with a basis $\{\phi_i\}_{i\ge 1}$, any $f : X \rightarrow \mathbb{R}$ can be represented as
\begin{equation}
  f = \sum_{i\ge 1} a_i\phi_i
\end{equation}
with the $a_i$ being some representation coefficients (in case of an orthonormal
basis, $a_i = \langle f,\phi_i \rangle$; in the general case, the coefficients are found by projecting the
function $f$ on the bi-orthonormal basis).
Due to the linearity of $T$,
\begin{equation}
T(f) = T\left (\sum_{i\ge1}a_i\phi_i \right ) = \sum_{i\ge1}a_i T(\phi_i)
\end{equation}
If the shape $Y$ is further equipped with a basis $\{\psi_j\}_{j\ge 1}$, then for every $i$ there exists coefficients $c_{ij}$ such that
\begin{equation}
  T(\phi_i) = \sum_{j\ge 1}c_{ij} \psi_j,
\end{equation}
and we can write
\begin{equation}
  T(f) = \sum_{i,j\ge 1}a_i c_{ij} \psi_j.
  \label{eq:Tf}
\end{equation}

Let us now assume finite sampling of $X$ and $Y$, with $m$ samples (for simplicity, we assume that the shapes are sampled at the same number of samples $m$. The extension to the case with a different number of samples is straightforward). The bases are represented as the $m \times n$ matrices $\bb{\Phi}$ and $\bb{\Psi}$ containing, respectively, $n$ discretized functions $\phi_i$ and $\psi_j$ as the columns. The functions $f$ and $g = T(f)$ can now be represented as $n$-dimensional vectors $\bb{f} = \bb{\Phi} \bb{a}$ and $\bb{g} = \bb{\Psi} \bb{b}$ with the coefficients $\bb{a}$ and $\bb{b}$.
Using this notation, Equation (\ref{eq:Tf}) can be rewritten as
\begin{equation}
\bb{\Psi} \bb{b} = T(\bb{\Phi}\bb{a}) = \bb{\Psi} \bb{C}^\Tr \bb{a};
\end{equation}
since $\bb{\Psi}$ is invertible, this simply means that
\begin{equation}
\bb{b}^\Tr = \bb{a}^\Tr \bb{C}.
\label{eq:b-aC}
\end{equation}
Thus, the $n \times n$ matrix $\bb{C}$ fully encodes the linear map $T$ between the functional spaces, and contains the coordinates in the basis $\bb{\Psi}$ of the mapped elements of the basis $\bb{\Phi}$.

\subsection{Point-to-point correspondence}
\label{sec:p2p}

%
Point-to-point correspondences assume that each point $i$ on $X$ corresponds to some point $j$ on $Y$.
In functional representation, this is equivalent to having
%
$\bb{C}$ that makes each row of $\bb{\Psi}\bb{C}^\Tr$ coincide with some row of $\bb{\Phi}$ \cite{ovsjanikov2012functional}.
In applications requiring point-to-point correspondence, given some $\bb{C}$, it can be converted into a point-to-point correspondence by thinking of $\bb{\Phi}$ and $\bb{\Psi}$ as $n$-dimensional points clouds, and orthogonal matrix $\bb{C}$ as a rigid alignment transformation between them.
%
%
%
This procedure is equivalent to iterative closest point (ICP) in $n$ dimensions \cite{ovsjanikov2012functional}, initialized with the given $\bb{C}_0$: first, for each row $i$ of $\bb{\Psi}\bb{C_0}^\Tr$, find the closest row $j^*_i$ in $\bb{\Phi}$ (this operation can be performed efficiently using approximate nearest neighbor algorithms).
Then, find orthonormal $\bb{C}$ minimizing $\sum_i \| \bb{\Phi}_{j^*_i} -\bb{\Psi}\bb{C}^\Tr \|_2$ and set $\bb{C}_0 = \bb{C}$.
This operation is repeated until convergence and can be performed efficiently over all the vertexes of $X$ and $Y$ using approximate nearest neighbor algorithms.

\subsection{Sparse modeling}

One of the main tools that will be used in this paper are \emph{sparse models}. In what follows, we give a very brief overview of this vast field, and refer the
reader to \cite{elad2010sparse} for a comprehensive treatise.
The central assertion of sparse modeling is that many families of signals (and later operations as here introduced) can be represented as a sparse linear combination in an appropriate domain, usually referred to as the \emph{dictionary}.
This can be written as $\bb{x} \approx \bb{D}\bb{z}$, where $\bb{x}$ denotes the signal, $\bb{D}$ the dictionary, and $\bb{z}$ the sparse vector of representation coefficients. The dictionary is often selected to be \emph{overcomplete}, i.e., with more columns than rows.

Finding the representation of a signal $\bb{x}$ in a given dictionary $\bb{D}$ is usually referred to as sparse representation \emph{pursuit} or \emph{sparse coding}. Among the variety of pursuit methods, we will focus on the so-called Lasso formulation \cite{tibshirani96} that finds $\bb{z}$ as the solution to the unconstrained convex program
\begin{eqnarray}
\min_{\bb{z}} \|\bb{x}-\bb{D}\bb{z}\|_2^2 + \lambda \|\bb{z}\|_1.
\label{eq:lasso}
\end{eqnarray}
The first term is the data fitting term, while the second term involving $\ell_1$ norm, $\|\bb{z}\|_1 = |z_1|+\hdots+|z_n|$,
promotes a sparse solution; the parameter $\lambda$ controls the relative importance of the latter.
Proximal splitting methods \cite{nesterov07} are among the most efficient and most frequently used numerical tools to solve problem (\ref{eq:lasso}); in Section~\ref{sec:numerics}, we present a variant of the proximal splitting algorithms for the solution of the pursuit problem arising in shape correspondence as detailed in the sequel.

In some cases, signals not admitting the simplistic model of element-wise sparsity still manifest more intricate types of \emph{structured} sparsity. In structured sparse models, the non-zero elements of $\bb{z}$ come in groups or, more generally, in hierarchies of groups. A common class of structured pursuit problems can be formulated as convex programs of the form
\begin{eqnarray}
\min_{\bb{z}} \|\bb{x}-\bb{D}\bb{z}\|_2^2 + \lambda \|\bb{z}\|_{1,2},
\end{eqnarray}
where the $\ell_{1,2}$ norm, $\|\bb{z}\|_{1,2} = \|\bb{z}_{1}\|_2 + \cdot +\|\bb{z}_{k}\|_2$, assumes that the vector $\bb{z}$ is decomposed into $k$ non-overlapping sub-vectors $\bb{z}_i$, and promotes group-wise sparse solutions (i.e., the solution will have a small number of groups with non-zero coefficients, but the sub-vectors representing each such non-zero group will be dense).

While structured sparse models enforce group structure of each representation vector independently, it is often useful to consider the structure shared by multiple vectors. \emph{Collaborative} sparse models operate on data matrices $\bb{X}$, in which each column corresponds to a data vector, and assert that the patterns of non-zero coefficients are shared across the corresponding representation vectors, $\bb{Z}$. This is achieved by solving a pursuit problem of the form
\begin{eqnarray}
\min_{\bb{Z}} \|\bb{X}-\bb{D}\bb{Z}\|_{\mathrm{F}}^2 + \lambda \|\bb{Z}\|_{2,1},
\end{eqnarray}
where the first term involving the Frobenius norm serves as the data fitting term, and the second term with the $\ell_{2,1}$ norm promotes row-wise sparsity of the solution. The $\ell_{2,1}$ norm is defined as
$\|\bb{Z}\|_{2,1} = \|\bb{z}_1^\Tr\|_{2} + \cdots + \|\bb{z}_m^\Tr\|_{2}$, where $\bb{z}_i^\Tr$ denotes the $i$-th row of $\bb{Z}$ (note the difference from the $\ell_{1,2}$ column-wise counterpart!).

In this paper, we use formulate the shape correspondence problem using a sparse model, and use sparse modeling tools to efficiently solve it.

\section{Sparse modeling of correspondences}

\begin{figure*}[t!]
   \begin{center}
   \includegraphics[width=.9\linewidth]{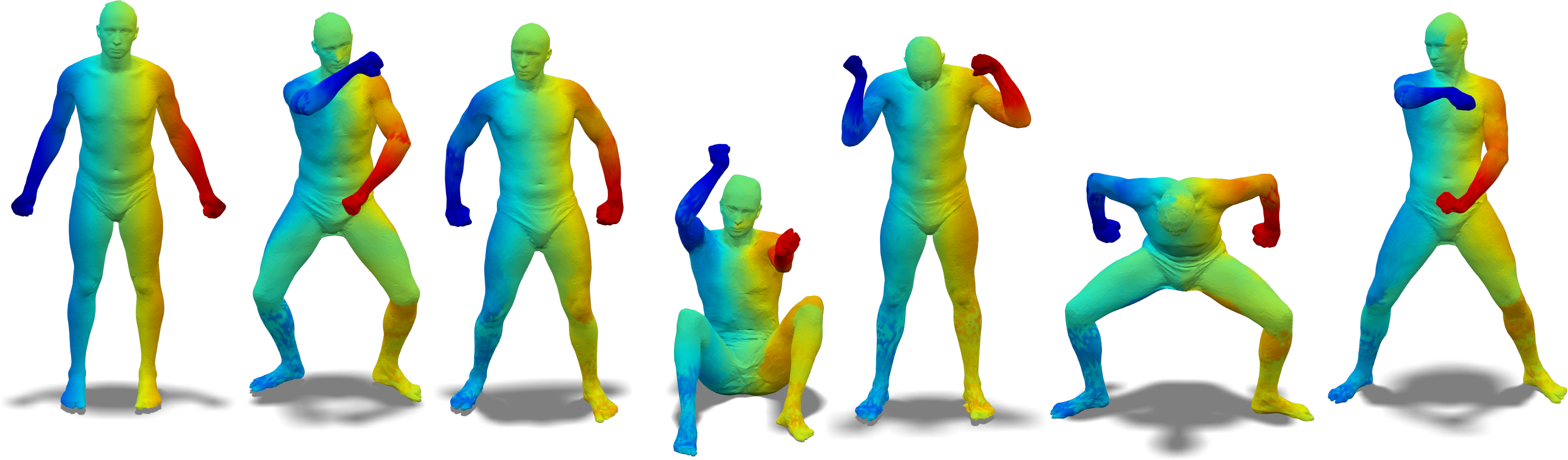}
   \caption{Dense point-to-point correspondences obtained between the left SCAPE human shape and various other poses. Corresponding points are marked with consistent colors. }
   \label{fig:scape}
   \end{center}
   \vspace{-1ex}
\end{figure*}

\begin{figure*}[t!]
   \begin{center}
   \includegraphics[width=.95\linewidth]{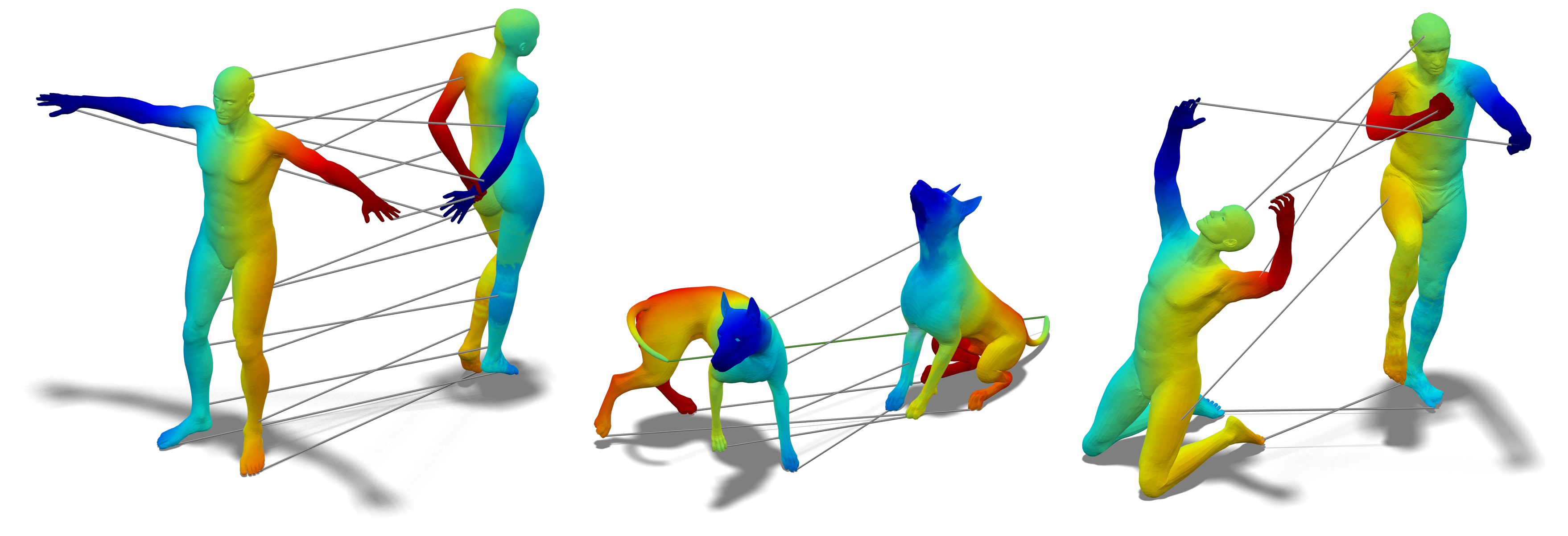} \\
   \includegraphics[width=\linewidth]{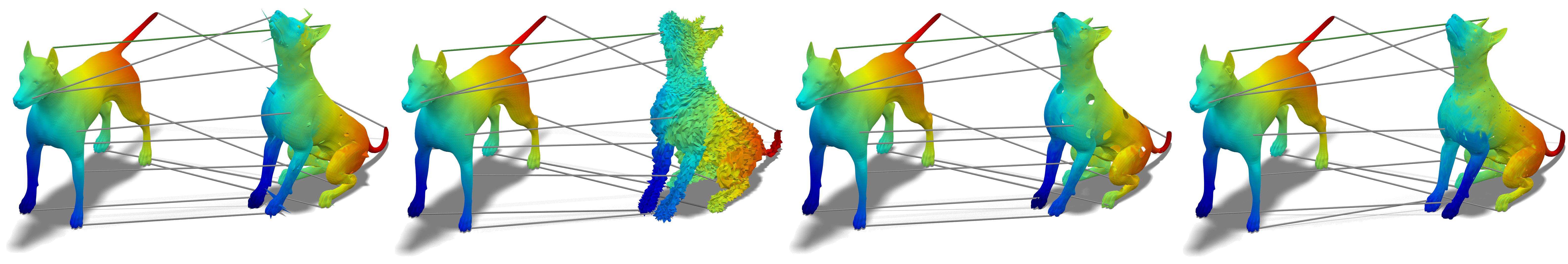}
   \caption{ First row: point-to-point correspondences obtained between different non-isometric shapes: male and female (left); two strongly non-isometric
   deformations of the dog shape from the TOSCA set (middle); TOSCA and SCAPE human shapes (right).
   Second row: Point-to-point correspondences obtained between SHREC shapes undergoing nearly isometric deformations and (from left to right)
   spike noise, Gaussian noise, and topological noise in the form of large and small holes.    }
   \label{fig:p2p}
   \end{center}
   \vspace{-1ex}
\end{figure*}


In case the shapes $X$ and $Y$ are isometric and the corresponding Laplace-Beltrami operators have simple spectra (no eigenvalues with multiplicity greater than one), the harmonic bases (Laplacian eigenfunctions) have a compatible behavior, $\psi_i = T(\phi_i)$ such that $c_{ij} = \pm \delta_{ij}$.
%
Choosing the 
discretized eigenfunctions of the Laplace-Beltrami operator 
as $\bb{\Phi}$ and $\bb{\Psi}$ causes every low-distortion correspondence being represented by a nearly diagonal, and therefore very sparse, matrix $\bb{C}$.

In practice, due to lack of perfect isometry and numerical inaccuracies, the diagonal structure of $\bb{C}$ is manifested for the first eigenfunctions corresponding to the small eigenvalues (low frequencies), and is gradually lost with the increase of the frequency (Figure~\ref{fig:funcorr}).
However, a correspondence with low metric distortion will usually be represented by a sparse $\bb{C}$.
We use this property to formulate the computation of correspondences in terms of a sparse representation pursuit problem.

Let us assume to have some {\em region} (or {\em feature}) {\em detection} process that given a shape $X$ produces a collection of functions $\{ f_i : X \rightarrow \mathbb{R} \}$ based on the intrinsic properties of the shape only.  For example, the $f_i$'s can be indicator functions of the maximally stable components (regions) of the shape \cite{litman:BB:11}. Since the process is intrinsic, given a nearly isometric deformation $Y$ or $X$, it should produce a collection of similar functions $\{ g_j : Y \rightarrow \mathbb{R} \}$.

To simplify the presentation, let us assume that the process is perfectly {\em repeatable} in the sense that it finds $q$ functions on $X$ and $Y$, such that for every $f_i$ there exists a $g_j = f_i \circ t$ related by the unknown correspondence $t$. We stress that the ordering of the $f_i$'s and $g_j$'s is \emph{unknown}, i.e., we do not know to which $g_j$ in $Y$ a $f_i$ in $X$ correspond. This ordering can be expressed by an unknown $q \times q$ permutation matrix $\bb{\Pi}$ (in Section~3.2, we consider the more general case when the number of functions detected on $X$ and $Y$ can be different, i.e., $\bb{\Pi}$ is non-square).

Representing the functions in the bases on each shape, we have $\bb{f}_i = \bb{\Phi} \bb{a}_i$ and $\bb{g}_j = \bb{\Psi} \bb{b}_j$.
Since each pair of corresponding $\bb{f}_i$ and $\bb{g}_j$ shall satisfy (\ref{eq:b-aC}), we can write
\begin{eqnarray}
\bb{\Pi} \bb{B} &=& \bb{A} \bb{C},
\label{eq:sparse-rel}
\end{eqnarray}
where $\bb{A}$ and $\bb{B}$ are the $q \times n$ matrices containing, respectively, $\bb{a}_i^\Tr$ and $\bb{b}_j^\Tr$ as the rows, and
$\pi_{ij} = 1$ if $\bb{a}_i$ corresponds to $\bb{b}_j$ and zero otherwise.

\subsection{Permuted sparse coding}

Note that in relation (\ref{eq:sparse-rel}), both $\bb{\Pi}$ and $\bb{C}$ are unknown, and solving for them is a highly ill-posed problem.
However, by recalling that the correspondence we are looking for should be represented by a nearly-diagonal $\bb{C}$, we formulate the following problem
\begin{eqnarray}
\min_{\bb{C},\bb{\Pi}} \frac{1}{2} \| \bb{\Pi} \bb{B} - \bb{A} \bb{C}\|_\mathrm{F}^2
+ \lambda \| \bb{W} \odot \bb{C} \|_1,
\label{eq:prob_sparse}
\end{eqnarray}
where the minimum is sought over $n \times n$ matrices $\bb{C}$ (capturing the correspondence $t$ between the shapes in the functional representation) and $q \times q$ permutations $\bb{\Pi}$ (capturing the correspondence between the detected regions on the shapes).
The first term containing the Frobenius norm can be interpreted as the data term, while the second term, containing the weighted $\ell_1$ norm promotes a sparse $\bb{C}$;  $\odot$ denotes element-wise multiplication, and the non-negative parameter $\lambda$ determines the relative importance of the penalty.
Small weights $w_{ij}$ in $\bb{W}$ are assigned close to the diagonal, while larger weights are selected for the off-diagonal elements. This promotes diagonal solutions.

The solution of (\ref{eq:prob_sparse}) can be obtained using alternating minimization iterating over $\bb{C}$ with fixed $\bb{\Pi}$,
and $\bb{\Pi}$ with fixed $\bb{C}$. Note that with fixed $\bb{\Pi}$, we can denote
$\bb{B}' = \bb{\Pi} \bb{B}$
and reduce problem (\ref{eq:prob_sparse}) to
\begin{eqnarray}
\min_{\bb{C}} \frac{1}{2} \| \bb{B}' - \bb{A} \bb{C}\|_\mathrm{F}^2
+ \lambda \| \bb{W} \odot \bb{C} \|_1,
\label{eq:prob_sparse1}
\end{eqnarray}
which resembles the Lasso problem frequently employed for the pursuit of sparse representations.
On the other hand, when $\bb{C}$ is fixed, we set
$\bb{A}' = \bb{A} \bb{C}$, reducing the optimization objective to
\begin{eqnarray}
\lefteqn{\| \bb{\Pi} \bb{B} - \bb{A}'\|_\mathrm{F}^2 \ = \  } \\
&& \mathrm{tr}\left( \bb{B}^\Tr \bb{\Pi}^\Tr \bb{\Pi} \bb{B} \right)
-2 \mathrm{tr}\left( \bb{B}^\Tr \bb{\Pi}^\Tr \bb{A}' \right) + \mathrm{tr}\left( \bb{A}^{\prime\Tr} \bb{A}' \right). \nonumber
\end{eqnarray}
Since $\bb{\Pi}$ is a permutation matrix, $\bb{\Pi}^\Tr \bb{\Pi} = \bb{I}$, and the only non-constant term remaining in the objective is the second linear term.
Problem (\ref{eq:prob_sparse}) thus becomes
\begin{eqnarray}
\max_{\bb{\Pi}} \mathrm{tr}\left( \bb{\Pi}^\Tr  \bb{E}  \right),
\end{eqnarray}
where
$\bb{E} = \bb{A}' \bb{B}^\Tr = \bb{A} \bb{C} \bb{B}^\Tr$
and the maximization is performed over permutation matrices.
To make it practically solvable, we allow $\bb{\Pi}$ to be a double-stochastic matrix, which yields the following linear assignment problem:
\begin{eqnarray}
\max_{\bb{\Pi} \ge \bb{0}} ~ \mathrm{vec}(\bb{E})^\Tr \mathrm{vec}(\bb{\Pi}) & \mathrm{s.t.} & \left\{ \begin{array}{l}
                                                                                                          \bb{\Pi} \bb{1} = \bb{1} \\
                                                                                                          \bb{\Pi}^\Tr \bb{1} = \bb{1}.
                                                                                                        \end{array} \right.
\label{eq:prob_linearass}
\end{eqnarray}

We refer to problem (\ref{eq:prob_sparse}) as to \emph{permuted sparse coding}, and propose to solve it by alternating the solution of the standard
sparse coding problem (\ref{eq:prob_sparse1}) and the solution of the linear assignment problem (\ref{eq:prob_linearass}).
The sparsity constraint has a regularization effect on this, otherwise extremely ill-posed, problem,
%
and the following strong property holds:
\begin{proposition}
The process alternating subproblems (\ref{eq:prob_sparse1}) and (\ref{eq:prob_linearass}) converges to a global minimizer of the permuted sparse coding problem (\ref{eq:prob_sparse}).
\end{proposition}
Due to lack of space, we provide the proof in the Appendix. This result means, among other, that despite the relaxation of the permutation matrix to a double-stochastic matrix in the assignment subproblem (\ref{eq:prob_linearass}), we are actually recovering a true permutation matrix. This follows from the total unimodularity of the constraints in (\ref{eq:prob_linearass}).

\subsection{Robust permuted sparse coding}

So far, we have assumed the existence of a bijective, albeit unknown, correspondence between the $f_i$'s and the $g_j$'s. In practice, the process detecting these functions (e.g., stable regions) is often not perfectly repeatable. In what follows, we will make a more realistic assumption that $q$ functions $f_i$ are detected on $X$, and $r$ functions $g_j$ detected on $Y$ (without loss of generality, $q \le r$), such that some $f_i$'s have no counterpart $g_j$, and vice versa. This partial correspondence can be described by a $q \times r$ partial permutation matrix $\bb{\Pi}$ in which now some columns and rows may vanish.

Let us assume that $s \le q$ $f_i$'s have corresponding $g_j$'s. This means that there is no correspondence between $r-s$ rows of $\bb{B}$ and $q-s$ rows of $\bb{A}$, and the relation
$\bb{\Pi} \bb{B} \approx \bb{A} \bb{C}$
holds only for an unknown subset of its rows. The mismatched rows of $\bb{B}$ can be ignored by letting some columns of $\bb{\Pi}$ vanish, meaning that the correspondence is no more surjective. This can be achieved by relaxing the equality constraint $\bb{\Pi}^\Tr \bb{1} = \bb{1}$ in (\ref{eq:prob_linearass}) replacing it with $\bb{\Pi}^\Tr \bb{1} \le \bb{1}$. However, dropping injectivity as well and relaxing $\bb{\Pi} \bb{1} = \bb{1}$ to $\bb{\Pi} \bb{1} \le \bb{1}$ would result in the trivial solution $\bb{\Pi} = \bb{0}$. To overcome this difficulty, we demand every row of $\bb{A}$ to have a matching row in $\bb{B}$, and absorb the $r-s$ mismatches in a row-sparse $q\times n$ outlier matrix $\bb{O}$ that we add to the data term of (\ref{eq:prob_sparse}).
This results in the following problem
\begin{eqnarray}
\min_{\bb{C},\bb{O},\bb{\Pi}} \frac{1}{2} \| \bb{\Pi} \bb{B} - \bb{A} \bb{C} - \bb{O} \|_\mathrm{F}^2 + \lambda \| \bb{W} \odot \bb{C} \|_1 + \mu \| \bb{O} \|_{2,1},
\label{eq:prob_sparse_robust}
\end{eqnarray}
which we refer to as \emph{robust permuted sparse coding}.
The last term involves the $\ell_{2,1}$ norm
\begin{eqnarray}
\| \bb{O} \|_{2,1}  &= & \sum_{i=1}^r \| \bb{o}^\Tr_i \|_2,
\end{eqnarray}
which can be thought of as the $\ell_1$ norm of the vector of the $\ell_2$ norms of the rows $\bb{o}_i^\Tr$ of $\bb{O}$.
The $\ell_{2,1}$ norm promotes row-wise sparsity, allowing to absorb the errors in the data term corresponding to the rows of $\bb{A}$ having no corresponding
rows in $\bb{B}$; the parameter $\mu \ge 0$ controls the amount of regularization.
%
%
The $q \times r$ matrix $\bb{\Pi}$ is searched over all injective correspondences.

As before, problem (\ref{eq:prob_sparse_robust}) is split into two sub-problems, one with the fixed permutation $\bb{\Pi}$,
\begin{eqnarray}
\min_{\bb{C},\bb{O}} \frac{1}{2} \| \bb{B}' - \bb{A} \bb{C} - \bb{O} \|_\mathrm{F}^2 + \lambda \| \bb{W} \odot \bb{C} \|_1 + \mu \| \bb{O} \|_{2,1},
\label{eq:prob_sparse_robust1}
\end{eqnarray}
with $\bb{B}' = \bb{\Pi} \bb{B}$, and the other one with the fixed $\bb{C}$,
\begin{eqnarray}
\max_{\bb{\Pi} \ge \bb{0}} ~ \mathrm{vec}(\bb{E})^\Tr \mathrm{vec}(\bb{\Pi}) & \mathrm{s.t.} & \left\{ \begin{array}{l}
                                                                                                          \bb{\Pi} \bb{1} = \bb{1} \\
                                                                                                          \bb{\Pi}^\Tr \bb{1} \le \bb{1},
                                                                                                        \end{array} \right.
\label{eq:prob_linearass_robust}
\end{eqnarray}
with $\bb{E} = \left(\bb{A} \bb{C} \right) \bb{B}^\Tr$.
Note that an injective correspondence is relaxed into a row-wise stochastic and column-wise sub-stochastic matrix $\bb{\Pi}$.
Proposition~1 simply extends to the robust formulation as well.

\section{Numerical solution}
\label{sec:numerics}

The solution of the robust permuted sparse coding problem (\ref{eq:prob_sparse_robust}) is reduced to alternating two relatively standard optimization problems, and there exist many readily available efficient numerical tools to solve them. For the sake of completeness, we provide a concise description of the involved numerics.

Problem (\ref{eq:prob_linearass_robust}), being a simple linear assignment problem, is solved using the Hungarian algorithm. As an alternative, linear programming can be employed.
To reduce the search space size, we disallow certain impossible permutations such as those relating regions with very distinct sizes.

In order to solve (\ref{eq:prob_sparse_robust1}), we use the family of forward-backward splitting algorithms \cite{nesterov07}
designed for solving unconstrained optimization problems in which the cost function can be split into the sum of two terms,
\begin{equation}
\min_{\bb{x}} h_1(\bb{x}) + h_2(\bb{x}),
\end{equation}
one, $h_1$, convex and differentiable with an $\alpha$-Lipschitz continuous gradient and another, $h_2$, convex extended real valued and possibly non-smooth. Clearly, problem (\ref{eq:prob_sparse_robust1}) falls in this category. 

The forward-backward splitting method with fixed constant step defines a series of iterates, $\{ \bb{x}^k\}_k$,
\begin{equation}
\bb{x}^{k+1} = \prox_{\alpha h_2}\left(\bb{x}^{k} - \frac{1}{\alpha} \nabla h_1(\bb{x}^{k}) \right),
\label{ec.iterates}
\end{equation}
where
\begin{eqnarray}
\prox_{\alpha h_2}(\bb{x}) = \mathrm{arg} \min_{\bb{u}} \, \| \bb{u} - \bb{x} ||_2^2 + \alpha  h_2(\bb{u})
\end{eqnarray}
denotes the proximal operator of $h_2$.
Many alternatives have been studied in the literature to improve the convergence rate of forward-backward splitting algorithms \cite{fista,nesterov07}. Accelerated versions reach quadratic convergence rates (the best possible for the class of first order methods). The discussion of theses methods is beyond of the scope of this paper.

In our case, the objective comprises a quadratic function $h_1 = \| \bb{B}' - \bb{A} \bb{C} - \bb{O} \|_\mathrm{F}^2$ and the non-smooth function $h_2 = \lambda \| \bb{W} \odot \bb{C} \|_1 + \mu \| \bb{O} \|_{2,1}$.
The proximal operator splits into two operators, one in $\bb{C}$ and another one in $\bb{O}$, both having closed forms. The proximal operator corresponding to the $\ell_1$ norm term is given by the weighted soft threshold function
\begin{eqnarray}
\prox_1(\bb{C}) &=& \max\left\{ |\bb{C}| - \frac{\lambda}{\alpha} \bb{W} \right\} \odot \mathrm{sign}(\bb{C}),
\end{eqnarray}
where the absolute value and the sign functions are applied element-wise.
The $i$-th row of the proximal operator corresponding to the $\ell_{2,1}$ norm term is given by
\begin{eqnarray}
(\prox_{2}(\bb{O}))_i = \max\left\{ \| \bb{o}_i^\Tr \|_2 - \frac{\mu}{\alpha}  \right\} \frac{ \bb{o}_i^\Tr }{ \| \bb{o}_i^\Tr \|_2 }.
\end{eqnarray}
The gradient of the quadratic data term with respect to $\bb{C}$ and $\bb{O}$ is given straightforwardly by
\begin{eqnarray}
\nabla_{\bb{C}} h_1 &=&  \bb{A}^\Tr \bb{A} \bb{C} + \bb{A}^\Tr \bb{O} - \bb{A}^\Tr \bb{B}'
\nonumber\\
\nabla_{\bb{O}} h_1 &=& \bb{O}  + \bb{A}\bb{C} - \bb{B}'.
\end{eqnarray}
The Lipschitz constant of the gradient determining the step size is bounded by the maximum eigenvalue
\begin{eqnarray}
\alpha & \le & \lambda_{\mathrm{max}}\left(
                               \begin{array}{cc}
                                 \bb{A}^\Tr \bb{A} & \bb{A}^\Tr \\
                                 \bb{I} & \bb{A} \\
                               \end{array}
                             \right).
\end{eqnarray}
Plugging the above expressions together into (\ref{ec.iterates}) yields the forward-backward splitting optimization
summarized in Algorithm~\ref{alg:F-B}.

\begin{algorithm}[tb]

\SetKwInOut{Input}{input}\SetKwInOut{Output}{output}

\Input{Data $\bb{B}', \bb{A}$; parameters $\lambda, \mu$; step size $\alpha$. }

\Output{Sparse matrix $\bb{O}$ and row-wise sparse outlier matrix $\bb{O}$}

Initialize $\bb{O}^0 = \bb{B}'$ and $\bb{C}^0 = \bb{0}$.

\For{k=1,2,\dots,until convergence}
{

$\bb{C}^{k+1} = \prox_1\left( (\bb{I} - \frac{1}{\alpha} \bb{A}^\Tr \bb{A}  ) \bb{C}^k - \frac{1}{\alpha} \bb{A}^\Tr(\bb{O}^k - \bb{B}')\right)$

$\bb{O}^{k+1} = \prox_2\left( (1 - \frac{1}{\alpha}) \bb{O}^k - \frac{1}{\alpha} ( \bb{A}\bb{C}^k - \bb{B}' ) \right)$

}
\caption{Forward-backward splitting method for the solution of (\ref{eq:prob_sparse_robust1}). \label{alg:F-B} }

\end{algorithm}

\begin{figure*}
   \begin{center}
   \includegraphics[width=\linewidth]{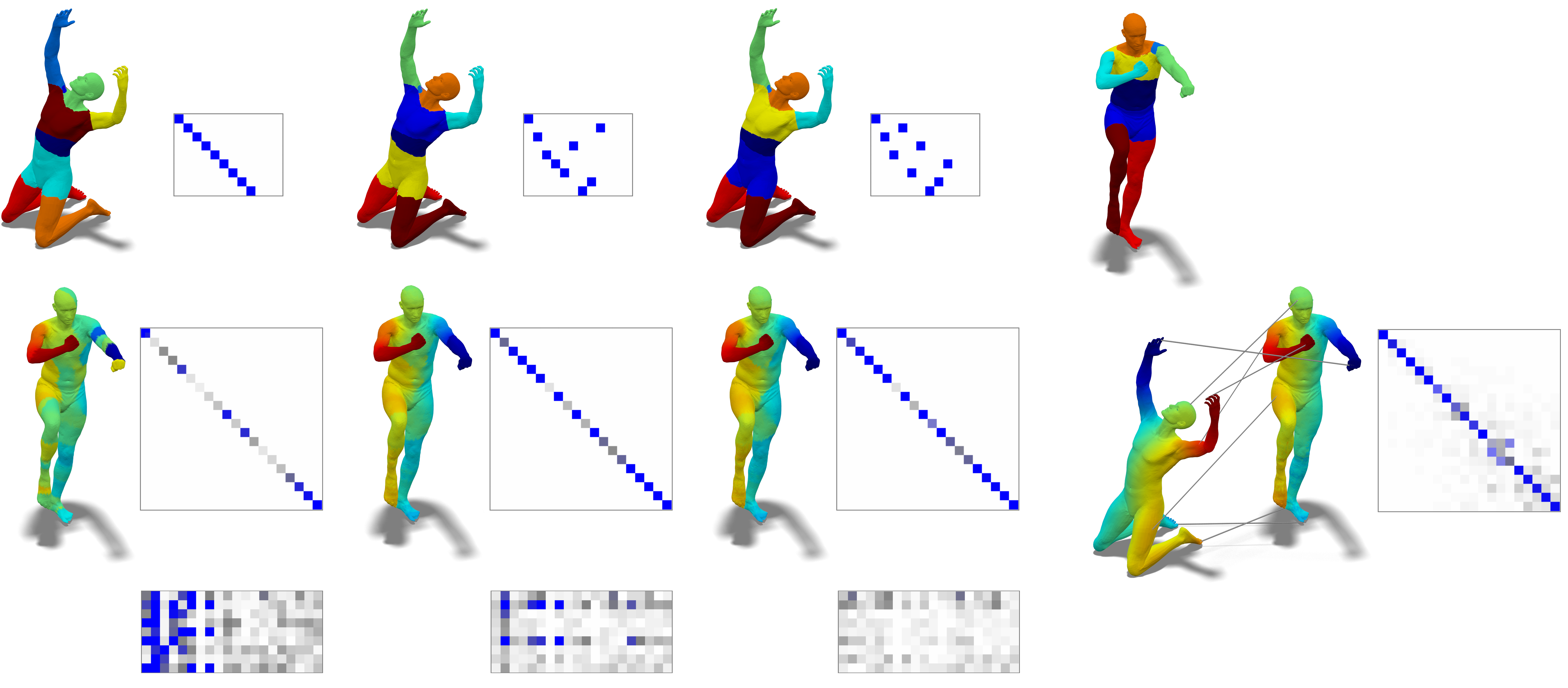}
   \caption{ Outer iterations of robust permuted sparse coding
   alternating the solution of the sparse representation purusit problem (\ref{eq:prob_sparse_robust1}) with
   the linear assignment problem (\ref{eq:prob_linearass_robust}). Three iterations, shown left-to-right, are required to achieve convergence.
   Depicted are the permutation matrix $\bb{\Pi}$ (first row), the correspondence matrix $\bb{C}$ (second row), and the outlier
   matrix $\bb{O}$ (last row). The resulting point-to-point correspondence and the correspondence matrix $\bb{C}$ refined using the ICP
   as described in Section~\ref{sec:p2p} are shown in the rightmost column. }
   \label{fig:iter}
   \end{center}
  \vspace{-1ex}
\end{figure*}

\section{Experimental results}

In order to evaluate our approach, we performed several experiments on data from the TOSCA \cite{bronstein2008ngn}, SHREC'11 \cite{shrec11}
and SCAPE \cite{scapedb}
datasets.
The TOSCA set contains high-quality ($10$K-$50$K vertices) synthetic triangular meshes of humans and animals in different poses with known ground truth correspondences between them.
SHREC'11 contains meshes from the TOSCA set undergoing simulated transformations.
The SCAPE set contains high-resolution ($12$K vertices) scans of a real human in different poses.

For each pair of shapes we calculated the MSERs using $6$-$10$ eigenfunctions and selected regions with areas of at least $5$-$10$\% of the total shape area, resulting in about $5-15$ detected regions (see Figure~\ref{fig:teaser}). These parameters were selected empirically for our data sets.

The segments of each shape were projected onto $20$ eigenfunctions and the corresponding $\bb{C}$ matrix was calculated by solving the sparse coding subproblem (\ref{eq:prob_sparse_robust1}) using an accelerated variant of the method described in Section~\ref{sec:numerics}.
%
The linear assignment subproblem (\ref{eq:prob_linearass}) was solved using the Hungarian method \cite{Kuhn1955}.
We initialized the permutation matrix with $\bb{\Pi} = \frac{1}{q} \bb{1}\bb{1}^\Tr$, and the correspondence matrix with $\bb{C} = \bb{0}$.
We observed a rapid convergence of the alternating minimization procedure in one or two iterations (see Figure~\ref{fig:iter} where for visualization purposed,  $\bb{\Pi}$ was initialized to identity). We found that the method consistently converged to the same solution regardless of the initialization.
Finally, after convergence of the alternating minimization, the resulting $\bb{C}$ was refined using the method described in Section~\ref{sec:p2p}. 

The robustness of the method is demonstrated in Figures~\ref{fig:tosca}--\ref{fig:p2p}; correct correspondences are computed
even when the shapes undergo non-isometric deformations and are contaminated by geometric or topological noise.
Figure~\ref{fig:benchmark} shows a quantitative evaluation and comparison of our algorithm to other correspondence algorithms on the SCAPE data set.
The evaluation was performed using the code and data from \cite{kim2011blended}. Our method outperforms existing methods while making less assumption and working only with intrinsic information.

\begin{figure}[tb]
   \begin{center}
   \includegraphics[width=\columnwidth]{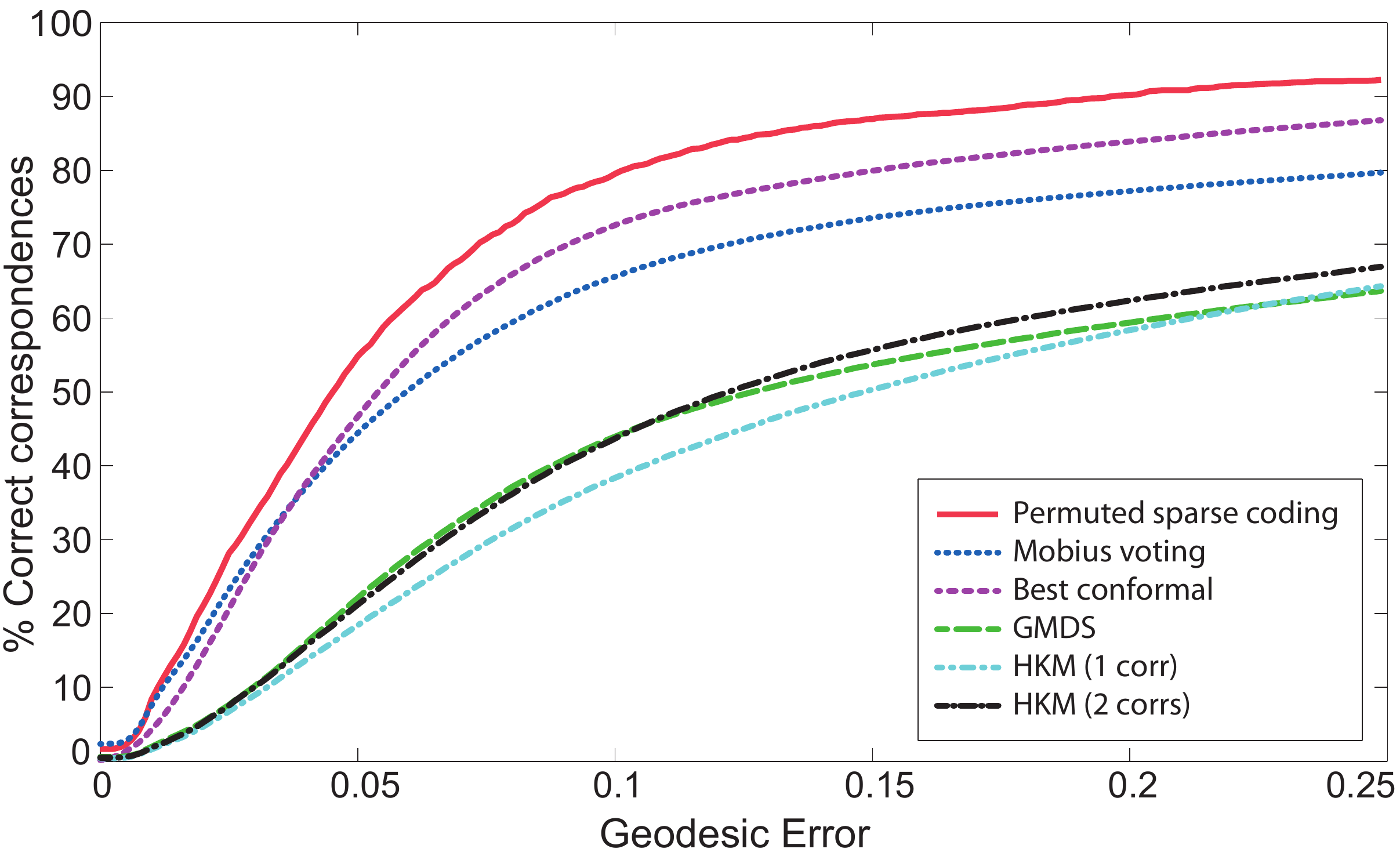}
   \caption{ Quantitative evaluation of the proposed shape correspondence algorithm and its comparison to
   other correspondence algorithms on the SCAPE shapes using the evaluation protocol from \cite{kim2011blended}. }
   \label{fig:benchmark}
   \end{center}
   \vspace{-1ex}
\end{figure}

The code used in the experiments was implemented in Matalb with parts written in C. The approximate nearest neighbor search in the ICP refinement step was accelerated using the FLANN library. The experiments were run on a 2.4GHz Intel Xeon CPU. End-to-end execution time varied from $10$ to $50$ seconds, with the detailed breakdown summarized in Table~\ref{tab:runtimes}.

\begin{table}[t]
\caption{\label{tab:runtimes} \small{Average runtime (in seconds) as a function of the shape size for different stages in the proposed method:
Basis -- harmonic basis computation; MSER -- region detection; Opt. -- alternating minimization procedure;
Ref. -- ICP-based refinement and point-to-point correspondence computation; Tot. -- total runtime. }}
   \begin{center}
      \begin{tabular}{lcccccc}
         \hline
         \hline
         \textbf{Vertices} & \textbf{Basis} & \textbf{MSER} & \textbf{Opt.} & \textbf{Ref.} & \textbf{Tot.}\\
         \hline
         ~~5K      &  0.53  &  0.61 & 7.80 &  1.41 & \textbf{10.35} \\
         10K       &  0.99  &  1.32 & 7.91 &  2.70 & \textbf{12.92} \\
         20K       &  2.03  &  3.58 & 7.91 &  5.52 & \textbf{19.04} \\
         50K       &  5.57  & 14.23 & 7.85 & 13.99 & \textbf{41.64} \\
         \hline
      \end{tabular}
   \end{center}
   \vspace{-2ex}
\end{table}

\section{Discussion and Conclusion}

In this paper, we posed the problem of finding intrinsic correspondence between near-isometric deformable shapes as a problem of sparse modeling.
Given only two set of regions in the two shapes with unknown correspondence, we are able to infer a dense correspondence between the shapes from two assumptions: that at least some of the regions in the two sets are corresponding; and that the shapes are nearly-isometric.
The latter assumption implies that in functional representation in harmonic bases the unknown correspondence between the shapes is modeled as a sparse nearly-diagonal matrix; the former assumption implies that there exists an unknown permutation that reorders the regions in corresponding order.
To find both the permutation and the correspondence, we formulate the novel permuted sparse coding problem and propose its efficient solution. An additional sparse coding term addressing outliers is added to the model for handling partial matching and formulated as the robust permuted sparse coding.

To the best of our knowledge, among other dense correspondence techniques, our method relies on the smallest amount of information (the ability to find some repeatable regions) and quite generic assumption (near-isometric shapes).
In particular, it allows us to use only a region detector without a feature descriptor to find a high-quality correspondence between two shapes.

We note that, as in \cite{ovsjanikov2012functional}, we explicitly assume that the shapes are nearly isometric, and that their Laplacians have simple spectrum. 
This assumption assures that the Laplacian eigenbases $\bb{\Phi}$ and $\bb{\Psi}$ have a compatible behavior, and as a result $\bb{C}$ has a nearly-diagonal structure.
If we try to relax the restriction on multiplicity, $\bb{C}$ will still be sparse, but with unknown sparse structure.
We can can still use our problem in this setting, imposing a different sparsity constraint on $\bb{C}$.

Relaxing the assumptions even more, we can depart from the near-isometric model, e.g. considering applications where one wishes to match shapes with roughly similar geometry but very different details (such as a horse and an elephant).
In such a generic setting, the Laplacian eigenbases may differ dramatically, and thus $\bb{C}$ have a non-sparse structure.
It is possible to incorporate the bases $\bb{\Phi}$ and $\bb{\Psi}$ as variables into our problem, and in addition to finding the permutation $\bb{\Pi}$ and correspondence $\bb{C}$ find also the bases in which $\bb{C}$ will have a diagonal structure. This problem is akin to dictionary learning used in the sparse modeling literature.
In future research, we will study such a generalization of our framework in hope to find correspondences between non-isometric shapes.
Another possible generalization of our problem is for finding correspondence between collections of shapes
\cite{nguyen2011optimization,kim2011blended}. 

Finally, it worthwhile noting that the novel structured sparse modeling techniques introduced in \cite{sprech12icml} provide an alternative to the optimization-based
pursuit by replacing the iterative proximal algorithm with a learned fixed-complexity feed-forward network. Approaching shape correspondence as a learning problem from this perspective seems a very attractive future research direction.

\section*{Appendix A -- Proof of Proposition~1}

The permuted sparse coding problem
\begin{eqnarray}
\min_{\bb{C} \in \mathbb{R}^{q\times q},\bb{\Pi} \in \mathbb{P}(q) } \frac{1}{2} \| \bb{\Pi} \bb{B} - \bb{A} \bb{C}\|_\mathrm{F}^2
+ \lambda \| \bb{W} \odot \bb{C} \|,
\label{eq:psc}
\end{eqnarray}
where $\mathbb{P}(q)$ denotes the space of $q \times q$ permutation matrices, is non-convex since the feasible set is non-convex.
However, by relaxing $\mathbb{P}(q)$ to the bigger space of $q\times q$ double-stochastic matrices, $\mathbb{S}(q) = \{ \bb{\Pi} \in \mathbb{R}^{q\times q}_+ : \bb{\Pi} \bb{1} = \bb{\Pi}^\Tr \bb{1} = \bb{1} \} \supset \mathbb{P}(q)$, we obtain the problem
\begin{eqnarray}
\min_{\bb{C} \in \mathbb{R}^{q\times q},\bb{\Pi} \in \mathbb{S}(q) } \frac{1}{2} \| \bb{\Pi} \bb{B} - \bb{A} \bb{C}\|_\mathrm{F}^2
+ \lambda \| \bb{W} \odot \bb{C} \|
\label{eq:psc_relaxed}
\end{eqnarray}
that is easily shown to be convex.
Fixing one of the variables at a time, the problem can be split into two subproblems: the sparse coding problem
\begin{eqnarray}
\min_{\bb{C}} \frac{1}{2} \| \bb{B}' - \bb{A} \bb{C}\|_\mathrm{F}^2
+ \lambda \| \bb{W} \odot \bb{C} \|_1
\label{eq:sc}
\end{eqnarray}
with $\bb{B}' = \bb{\Pi}\bb{B}$, and the linear assignment problem
\begin{eqnarray}
\max_{\bb{\Pi} \in \mathbb{S}(q) } ~ \mathrm{vec}(\bb{E})^\Tr \mathrm{vec}(\bb{\Pi})
\label{eq:linass}
\end{eqnarray}
with $\bb{E} = \bb{A}' \bb{B}^\Tr = \bb{A} \bb{C} \bb{B}^\Tr$.
These two subproblems can be viewed as minimizing the objective function of (\ref{eq:psc_relaxed}) with respect to two blocks of coordinates,
$\bb{C}$ and $\bb{\Pi}$. The minimization process alternating between the solution of (\ref{eq:sc}) and (\ref{eq:linass}) can be therefore
regarded as block-coordinate descent.

We use an instance of Theorem~4.1 in \cite{tseng2001convergence} stating that block-coordinate descent is guaranteed to converge
to a global minimizer of a non-differentiable convex function.
The conditions of the theorem are satisfied by the objective and the constraints of (\ref{eq:psc_relaxed}). Note that since we do not prove strict convexity of the latter objective, a global minimizer is not necessarily unique.

Let us now have a closer look at the linear assignment problem (\ref{eq:linass}) that can be cast as the linear program
\begin{eqnarray}
\min_{\bb{\pi} \in \mathbb{R}^{q^2}_+ } ~ \bb{e}^\Tr \bb{\pi} &\mathrm{s.t.}&  \bb{Q} \bb{\pi} = \bb{1},
\label{eq:linprog}
\end{eqnarray}
where the variable vector $\bb{\pi} = \mathrm{vec}(\bb{\Pi})$ is the column stack of the assignment matrix,
the cost vector is given by $\bb{e} = -\mathrm{vec}(\bb{E})$, and the constraint matrix can be expressed using the Kr\"{o}necker product notation as
\begin{eqnarray}
\bb{Q} &=& \left(
             \begin{array}{c}
               \bb{1}^\Tr \otimes \bb{I} \\
               \bb{I} \otimes \bb{1}^\Tr  \\
             \end{array}
           \right).
           \label{eq:Q}
\end{eqnarray}
Here, $\bb{I}$ stands for the $q \times q$ identity matrix, and $\bb{1}^\Tr$ for the $1\times q$ vector of ones.
In Lemma~1 below, we prove that $\bb{Q}$ is totally unimodular, i.e., any of its square submatrices has the determinant of $0$ or $\pm 1$. This property guarantees that
the linear program has a solution with integer coordinates. Since $\mathbb{S}(q) \cap \mathbb{Z}^{q \times q} = \mathbb{P}(q)$,
this guarantees that the linear assignment problem (\ref{eq:linass}) is actually equivalent to the binary assignment problem
\begin{eqnarray}
\max_{\bb{\Pi} \in \mathbb{P}(q) } ~ \mathrm{vec}(\bb{E})^\Tr \mathrm{vec}(\bb{\Pi}).
\label{eq:linass1}
\end{eqnarray}
Combined this result with the convergence of the block-coordinate descent to a global minimizer of (\ref{eq:psc_relaxed}), we can guarantee that in the obtained
minimizer $\bb{\Pi}$ actually belongs to $\mathbb{P}(q)$. This implies that the block-coordinate descent converges to a global minimizer of (\ref{eq:psc}).

\begin{lemma}
The matrix $\bb{Q}$ in (\ref{eq:Q}) is totally unimodular.
\end{lemma}

\begin{proof}
The matrix $\bb{Q}$ can be constructed as the sub-matrix of a bigger matrix
\begin{eqnarray}
\bb{\overline{Q}} &=& \left(
             \begin{array}{c}
               \bb{1}^\Tr  \\
               \bb{I}   \\
             \end{array}
           \right) \otimes \left(
             \begin{array}{c}
               \bb{I}  \\
               \bb{1}^\Tr   \\
             \end{array}
           \right).
\end{eqnarray}
Hence, total unimodularity of $\bb{\overline{Q}}$ implies total unimodularity of $\bb{Q}$.
Since total unimodularity is preserved by the Kr\"{o}necker product, it is sufficient to show that each of the two
Kr\"{o}necker factors are totally unimodular. We will limit the discussion to the second factor; very similar arguments apply to the first one.

The matrix
\begin{eqnarray}
\bb{R} &=&  \left(
             \begin{array}{c}
               \bb{I}  \\
               \bb{1}^\Tr   \\
             \end{array}
           \right)
\end{eqnarray}
comprises two components: the identity matrix $\bb{I}$ and the row vector $\bb{1}^\Tr$, both of which are totally unimodular. For each $k \times k$ square submatrix $\bb{R}'$ of $\bb{R}$, we distinguish between the following three cases: $\bb{R}'$ containing only elements of $\bb{I}$; $\bb{R}'$ containing only elements of $\bb{1}^\Tr$; and $\bb{R}'$ containing elements of both components. In the two former cases, $\det \bb{R}' \in \{0,\pm 1\}$ since $\bb{I}$ and $\bb{1}^\Tr$ are totally unimodular. In the latter case, the submatrix has the form
\begin{eqnarray}
\bb{R}' &=&  \left(
             \begin{array}{cc}
              \bb{O} & \bb{I} \\
              \bb{O} & \bb{O}  \\
               \bb{1} & \bb{1}^\Tr    \\
             \end{array}
           \right),
\end{eqnarray}
where the $m \times m$ ($m < k$) identity matrix $\bb{I}$ is surrounded by zeros and concatenated to a row of ones on the bottom. For $m<k-1$, the submatrix contains at least one row of zeros and therefore $\det \bb{R}' = 0$. For $m=k-1$, $\bb{R}'$ assumes the form
\begin{eqnarray}
\bb{R}' &=&  \left(
             \begin{array}{cc}
               \bb{0} & \bb{I}  \\
               1 & \bb{1}^\Tr     \\
             \end{array}
           \right),
\end{eqnarray}
where $\bb{I}$ is the $(k-1) \times (k-1)$ identity matrix, $\bb{0}$ is the $(k-1) \times 1$ vector of zeros, and $\bb{1}^\Tr$ is the $1 \times (k-1)$ vector of ones.
Using the properties of the determinant, we obtain
\begin{eqnarray}
\det \bb{R}' &=& \det (\bb{0} \bb{1}^\Tr - \bb{I}) = (-1)^{k-1}.
\end{eqnarray}
Hence, $\bb{R}$ is totally unimodular.
\end{proof} 

\bibliographystyle{alpha}
\bibliography{optimization,shapes}

\vfill

\end{document}